\documentclass{wstmp}
\begin{document}
\title{
Abstract tubes associated with perturbed polyhedra with applications to multidimensional normal probability computations
}

\author{S. KURIKI}

\address{The Institute of Statistical Mathematics \\
10-3 Midoricho, Tachikawa, Tokyo 190-8562, Japan \\
E-mail: kuriki@ism.ac.jp}

\author{T. MIWA}

\address{
National Institute for Agro-Environmental Sciences \\
3-1-3 Kannondai, Tsukuba 305-8604, Japan \\
E-mail: miwa@niaes.affrc.go.jp}

\author{A. J. HAYTER}

\address{
Daniels College of Business, University of Denver \\
2101 S.\,University Blvd., Denver, CO 80208-8921, USA \\
E-mail: Anthony.Hayter@du.edu}

\begin{abstract}
Let $K$ be a closed convex polyhedron defined by a finite number of
linear inequalities.  In this paper
we refine the theory of abstract tubes (Naiman and Wynn, 1997)
associated with $K$ when $K$ is perturbed.  In particular,
we focus on the perturbation that is lexicographic and in an outer direction.
An algorithm for constructing the abstract tube by means of
linear programming and its implementation are discussed.
Using the abstract tube for perturbed $K$
combined with the recursive integration technique
proposed by Miwa, Hayter and Kuriki (2003),
we show that the multidimensional normal probability
for a polyhedral region $K$ can be computed efficiently.
In addition, abstract tubes and the distribution functions of 
studentized range statistics are exhibited as numerical examples.
\end{abstract}

\keywords{
Abstract Tube; Inclusion-Exclusion Identity; Lexicographic Method; Linear Programming; Multiple Comparisons; Perturbation; Studentized Range.}

\bodymatter

\section{Introduction}
\label{sec:introduction}

Let $A=(a_1,\ldots,a_m)$ be an $n\times m$ matrix such that
$a_i\ne 0$ for all $i$, and let $b=(b_1,\ldots,b_m)^\top\in\mathbb{R}^m$ be
an $m\times 1$ constant vector.
Define a closed convex polyhedron by
\[
 K=\{ x\in\mathbb{R}^n \mid A^\top x \le b \}.
\]
Throughout the paper, $K$ is assumed to be nonempty.
Suppose that $x\in\mathbb{R}^n$ is a random vector distributed
as an $n$-dimensional standard normal distribution $N_n(0,I_n)$.
The primary motivation and hence one of the purposes of this study is
to evaluate the $n$-dimensional normal probability
\begin{equation}
\label{PK}
 P(K)=P_n(K):=\Pr(x\in K).
\end{equation}

When $m=n$ and $A$ is non-singular, $K$ is a cone referred to as
the simple cone formed as the intersection of
$n$ half spaces located in the general positions
(Barvinok, 2002\cite{Barvinok02}).
In this paper, we do not require that the apex of cone lies at the origin.
In the case where $m<n$ and $\mathrm{rank}(A)=m$, $K$ is decomposed as
$K=K_1\oplus \mathbb{R}^{n-m}$
with $K_1$ an $m$-dimensional simple cone,
where ``$\oplus$'' means the orthogonal direct sum.
In this case
the probability (\ref{PK}) is obtained as $P_n(K)=P_m(K_1)$.
For such a simple cone $K$, Miwa, Hayter and Kuriki (2003)\cite{Miwa-etal03}
proposed a recursive integration algorithm to evaluate the probability $P(K)$
by generalizing the idea of Abrahamson (1964)\cite{Abrahamson64}.
This algorithm is practically useful for dimensions $n$ up to 20,
and is available in the R library \texttt{mvtnorm}\cite{mvtnorm,Genz-Bretz09}.

In this paper, using this integration algorithm as a building block,
we demonstrate that the probability $P(K)$ for any polyhedron $K$
can be computed by means of 
the abstract tube ideas proposed by Naiman and Wynn (1997)\cite{Naiman-Wynn97}.
This method provides a sophisticated version of
the inclusion-exclusion identity,
and divides the complimentary set $K^c$ into signed cones.
By calculating the normal probability for each cone
and summing them up, we obtain $P(K)=1-P(K^c)$.
However, as we shall see later,
the resulting cones are not necessarily simple, and
we need to introduce a perturbation into the linear inequality system.
We focus on a particular type of perturbation,
and refine the theory of abstract tubes 
for a perturbed inequality system in this setting.
Some statements are simply taken from
Naiman and Wynn (1997)\cite{Naiman-Wynn97},
but some of them are novel.
For example, we show that under our perturbation
the perturbed system always defines an abstract tube
instead of a ``weak abstract tube''.

Moreover, we show that the abstract tube for the perturbed system
can be constructed by means of a standard linear programming technique.
We propose that such an algorithm should be used,
and we make some remarks that should be
incorporated into its implementation.
The prototype R program is available from the authors.

The construction of the paper is as follows.
In Section \ref{sec:perturbation}, the theory of
the abstract tube for a perturbed system is investigated.
In Section \ref{sec:construction},
the algorithm and its implementation for constructing
the abstract tube is discussed.
In Section \ref{sec:studentized}, we exhibit numerical examples.
For Tukey's studentized range statistics,
the abstract tubes are constructed in various settings, and their
distribution functions are computed.

\section{Abstract tubes for a perturbed inequality system}
\label{sec:perturbation}

In this section we investigate the abstract tube
associated with a closed convex polyhedron.
Because of the reason stated later, we focus on the case where
the inequality system defining the polyhedron is perturbed.
We propose a lexicographic perturbation in an outer direction
 (see (\ref{Ke})) and refine the theoretical results of
Naiman and Wynn (1997)\cite{Naiman-Wynn97} in this setting.

Define $m$ half spaces
\[
 H_i = \{ x\in\mathbb{R}^n \mid a_i^\top x \le b_i \},\ \ i=1,\ldots,m.
\]
The complement and the boundary of $H_i$ are denoted by
\[
 H_i^c = \{ x\in\mathbb{R}^n \mid a_i^\top x > b_i \} \quad\mbox{and}\quad
 \partial H_i = \{ x\in\mathbb{R}^n \mid a_i^\top x = b_i \},
\]
respectively.  Since
$P(K) = P(\cap_{i=1}^m H_i) = 1-P(\cup_{i=1}^m H_i^c)$,
we can evaluate $P(\cup_{i=1}^m H_i^c)$ instead of $P(\cap_{i=1}^m H_i)$.

The family of the complements of the half spaces is denoted by
\[
 \mathcal{H} = \{ H_i^c \mid i=1,\ldots,m \}.
\]
Let
\[
 F_i = \partial H_i \cap K,\ \ i=1,\ldots,m.
\] 
If $F_i\ne\emptyset$, $F_i$ is a face of $K$.
Define a family of subsets of the indices $\{1,\ldots,m\}$ by
\[
 \mathcal{F} = \{ J\subseteq \{1,\ldots,m\} \mid
 J \ne \emptyset,\,\cap_{i\in J} F_i \ne \emptyset \}.
\]
Then, $\mathcal{F}$ forms an (abstract) simplicial complex.
$\{ \cap_{i\in J} F_i \mid J\in\mathcal{F} \}$ is the family of
all faces (vertex, edge, ..., facet) of $K$.
Note that $\cap_{i\in J_1} F_i=\cap_{i\in J_2} F_i$ can happen
even though $J_1\ne J_2$.

Let
\[
 \mathbf{1}_S(x) = \begin{cases}
 1, & \mbox{if $x\in S$}, \\
 0, & \mbox{otherwise},
\end{cases}
\]
be the indicator function of the set $S$.
The following is Theorem 2 of Naiman and Wynn (1997)\cite{Naiman-Wynn97}.
\begin{proposition}
\label{prop:at}
The pair $(\mathcal{H},\mathcal{F})$ is an abstract tube in the sense that
\begin{equation}
\label{at}
 \mathbf{1}_{\cup_{i=1}^m H_i^c}(x) = \sum_{J\in\mathcal{F}}
 (-1)^{|J|-1} \mathbf{1}_{\cap_{i\in J} H_i^c}(x)
 \quad\mbox{for all $x\in\mathbb{R}^n$},
\end{equation}
where $|J|$ is the cardinality of the set $J$.
\end{proposition}

Assuming that $x\in\mathbb{R}^n$ is a standard Gaussian vector
distributed as $N_n(0,I_n)$,
and taking the expectation of (\ref{at}), we obtain the formula
\[
  1-P(K) = \sum_{J\in\mathcal{F}} (-1)^{|J|-1} P(\cap_{i\in J} H_i^c).
\]
In the right side,
if $\cap_{i\in J} H_i^c$ forms an $n$-dimensional simple cone or the
direct sum of a simple cone and a linear subspace,
we can use the algorithm of Miwa et al.\ (2003)\cite{Miwa-etal03}
for evaluating the probability $P(\cap_{i\in J} H_i^c)$.
However, this is not always the case, and for it to be the case
the condition of ``general position'' is required.
Before defining this notion,
we explain two examples of abstract tubes that do not meet this condition.

\newcommand{\Hspace}{\hspace*{15mm}}
\begin{example}
\label{ex:pyramid}
Consider a polyhedral cone called a ``pyramid'' $K\subset\mathbb{R}^3$
consisting of $(x_1,x_2,x_3)$ satisfying
\begin{align*}
-x_1 - x_2 + x_3 \le 1 & \Hspace [1] \\
-x_1 + x_2 + x_3 \le 1 & \Hspace [2] \\
+x_1 + x_2 + x_3 \le 1 & \Hspace [3] \\
+x_1 - x_2 + x_3 \le 1 & \Hspace [4]
\end{align*}
The boundaries of these half spaces
are not in the general position in the sense that
4 hyperplanes share a point $(0,0,1)$ in $\mathbb{R}^3$.
The simplicial complex is shown to be
\begin{equation}
\label{F}
\mathcal{F} = \{ 1,\, 2,\, 3,\, 4,\,
  12,\, 13,\, 14,\, 23,\, 24,\, 34,\, 123,\, 124,\, 134,\, 234,\, 1234 \}.
\end{equation}
In the expression above ``134'' means $\{1,3,4\}$ for example,
and this convention is adopted throughout the paper.
Applying this to the Proposition \ref{prop:at}, we can decompose $1-P(K)$ into
$|\mathcal{F}|=15$ terms.
However, this decomposition is not appropriate for our purpose
because the term $1234$ stands for a non-simple cone (pyramid), and hence
this decomposition does not simplify our problem.
\end{example}

\begin{example}
\label{ex:redundant}
Consider the following system containing a redundant inequality:
\begin{align*}
 +x_1 - x_2 + 0 x_3 \le 0 & \Hspace [1] \\
 -x_1 - x_2 + 0 x_3 \le 0 & \Hspace [2] \\
 0 x_1 - x_2+ 0 x_3 \le 0 & \Hspace [3]
\end{align*}
The third inequality is redundant.
This example is essentially 2-dimensional, and
3 boundaries (lines) meet at the origin in $\mathbb{R}^2$.
The simplicial complex is shown to be
\[
 \mathcal{F} = \{ 1,\, 2,\, 3,\, 12,\, 13,\, 23,\, 123 \}.
\]
This decomposition is also unsuitable
because the 3-dimensional term $123$
appears although this example is essentially 2-dimensional.
\end{example}

To avoid unfavorable events such as have appeared in these examples
we need the assumption that $\{\partial H_i\}$ is
in the general position defined below.

\begin{definition}
The set of $m$ hyperplanes in $\mathbb{R}^n$, for example,
$\{\partial H_i\subset\mathbb{R}^n \mid i=1,\ldots,m\}$,
is said to be in the general position when there does not exist
$J\subseteq\{1,\ldots,m\}$ such that $\cap_{i\in J} \partial H_i$
contains a $\max\{n+1-|J|,0\}$ dimensional affine subspace.
\end{definition}

\begin{remark}
This definition is equivalent to that in Stanley (2007)\cite{Stanley07}.
The definition by Naiman and Wynn (1997)\cite{Naiman-Wynn97}
is weaker than ours.
In their definition, when there does not exist $J\subseteq\{1,\ldots,m\}$,
$|J|=n+1$, $\cap_{i\in J} \partial H_i\ne\emptyset$,
the set of hyperplanes is said to be in the general position.
Example \ref{ex:redundant} is in the general position
in their definition,
whereas it is not in the general position in our definition.
\end{remark}

When a given $\{\partial H_i\}$ is not in the general position,
we can apply an infinitesimal perturbation to rearrange the system
into the general position.
In this paper we restrict our attention to the lexicographic perturbation
in an outer direction proposed below.

For $b=(b_i)\in\mathbb{R}^m$ and $\varepsilon>0$, define 
\[
 b(\varepsilon) = (b_i(\varepsilon)) = (b_i + \varepsilon^i)
 = (b_1 + \varepsilon^1,\ldots,b_m + \varepsilon^m)^\top
\]
($\varepsilon^i$ is $\varepsilon$ to the power $i$),
and
\begin{equation}
\label{Ke}
 K(\varepsilon) = \{ x\in\mathbb{R}^n \mid A^\top x \le b(\varepsilon) \}.
\end{equation}
Define $H_i(\varepsilon)$, $F_i(\varepsilon)$ and $\mathcal{F}(\varepsilon)$
similarly.
Note that $K=K(0)$, $H_i=H_i(0)$, $F_i=F_i(0)$ and
$\mathcal{F}=\mathcal{F}(0)$.

The following two lemmas, Lemmas \ref{lem:general} and \ref{lem:number},
are refinements of Lemma 1 of Naiman and Wynn (1997)\cite{Naiman-Wynn97}
under this lexicographic perturbation
and with the stronger definition of the general position.
\begin{lemma}
\label{lem:general}
For all sufficiently small $\varepsilon>0$,
\vspace*{-1.5mm}
\begin{itemize}
\item[(i)]
the family of hyperplanes
$\{\partial H_i(\varepsilon) \mid i=1,\ldots,m\}$ is in the general position,
and
\smallskip
\item[(ii)]
$\mathcal{F}(\varepsilon)$ does not depend on $\varepsilon$.
Write $\mathcal{F}(0+)=\mathcal{F}(\varepsilon)$.
$\mathcal{F}(0+)$ is a simplicial complex.
\end{itemize}
\end{lemma}

\begin{proof}
(i)
Fix $J\subseteq\{1,\ldots,m\}$.  Consider first the case $|J|\le n$.
Let $A_J=(a_i)_{i\in J}$ be $n\times |J|$, and let
$b_J(\varepsilon)=(b_i(\varepsilon))_{i\in J}$ be $|J|\times 1$.
Suppose that $\cap_{i\in J} \partial H_i(\varepsilon)$ contains an $n+1-|J|$
dimensional affine subspace $\{x = C d + x_0 \mid d\in\mathbb{R}^{n+1-|J|}\}$,
where $C=(c_1,\ldots,c_{n+1-|J|})$ is an $n\times (n+1-|J|)$ matrix
with linearly independent column vectors.
Then $A_J^\top (C d + x_0) - b_J(\varepsilon) = 0$ for all $d$, and hence
$A_J^\top C =0$ and $A_J^\top x_0 - b_J(\varepsilon) = 0$ hold.
Since $\mathrm{rank}(C)=n+1-|J|$, we see that $\mathrm{rank}(A_J)\le |J|-1$.
Hence $A_J$ can be written as $A_J=E F$, where
$E$ is $n\times r$ and $F$ is $r\times |J|$
with $r=\mathrm{rank}(A_J)\le |J|-1$.
Substituting this, we have 
\[
 0 = A_J^\top x_0 - b_J(\varepsilon) = F^\top (E^\top x_0) - b_J(\varepsilon)
 = (F^\top, -b_J(\varepsilon)) \begin{pmatrix}E^\top x_0 \\ 1 \end{pmatrix}.
\]
Therefore,
\begin{align}
 f_J(\varepsilon)
& := \det
 \biggl\{ \begin{pmatrix} F \\ -b_J(\varepsilon)^\top \end{pmatrix}
 (F^\top, -b_J(\varepsilon)) \biggr\} \nonumber \\
& = \det(F F^\top) \times b_J(\varepsilon)^\top
 \bigl\{ I - F^\top (F F^\top)^{-1} F \bigr\} \, b_J(\varepsilon) = 0.
\label{fJ}
\end{align}
Since $\mathrm{rank}(F)=r$,
$f_J(\varepsilon)$ is a polynomial in $\varepsilon$
of at least degree $2r\,(\ge 2)$.  Then
\[
 \delta_J=\sup\{\delta>0 \mid f_J(\varepsilon)>0,\ %
 \forall\varepsilon\in (0, \delta) \}>0,
\]
since the number of zeros of $f_J$ is finite.
For all $\varepsilon\in (0,\delta_J)$,
$\cap_{i\in J} \partial H_i(\varepsilon)$ contains no $(n+1-|J|)$-dimensional
affine subspace.

For the case $|J|\ge n+1$, assume that
$x_0\in\cap_{i\in J} \partial H_i(\varepsilon)$ exists.
Since $r=\mathrm{rank}(A_J)\le n\le |J|-1$,
we can follow the proof for $|J|\le n$ above.
We define the polynomial $f_J(\varepsilon)$ in (\ref{fJ}),
and conclude that for all $\varepsilon\in (0,\delta_J)$,
$\cap_{i\in J} \partial H_i(\varepsilon)=\emptyset$.

Let $\delta_*=\min_J \delta_J>0$.
For all $\varepsilon\in (0,\delta_*)$,
$\{\partial H_i(\varepsilon) \mid i=1,\ldots,m\}$ is in the general position.

(ii) The proof for (i) implies that for a given $J\subseteq\{1,\ldots,m\}$
the feasibility or infeasibility is unchanged
for all $\varepsilon\in(0,\delta_*)$. 
\end{proof}

\begin{corollary}
\label{co:li}
For each $J\in\mathcal{F}(0+)$,
\vspace*{-1.5mm}
\begin{itemize}
\item[(i)]
$\{ a_i \mid i\in J\}$ is linearly independent, and

\smallskip
\item[(ii)]
the number of elements $|J|$ is at most $\mathrm{rank}(A)$.
\end{itemize}
\end{corollary}

\begin{proof}
(i)
Note first that $\cap_{i\in J} \partial H_i$ contains a point
 (i.e., 0-dimensional affine subspace), and because
$\{ \partial H_i \mid i\in J\}$ is in the general position
(Lemma \ref{lem:general}, (i)), it should be that
$\max\{n+1-|J|,0\}\ge 1$ or $|J|\le n$.

Suppose that $\{ a_i \mid i\in J\}$, $J\in\mathcal{F}(0+)$,
is linearly dependent.
Let $L=\mathrm{span}\{a_i \mid i\in J\}$ and
$N=\{x\mid a_i^\top x=0,\,i\in J\}$.
Then $|J|\ge\dim(L)+1=n-\dim(N)+1$ and hence $\dim(N)\ge n+1-|J|$.
On the other hand, $\cap_{i\in J} \partial H_i\supseteq N + x_0$,
where $x_0\in \cap_{i\in J} \partial H_i$.
This contradicts the fact
that $\{ \partial H_i \mid i\in J\}$ is in the general position.

(ii)
The number of linearly independent vectors $\{ a_i \mid i\in J \}$
must be at most $\dim\mathrm{span}\{a_i\}=\mathrm{rank}(A)$. 
\end{proof}

\begin{lemma}
\label{lem:number}
Let $\mathcal{F}$ and $\mathcal{F}(0+)$ be simplicial complexes
defined in Proposition \ref{prop:at} and Lemma \ref{lem:general}, respectively.
Then $\mathcal{F}\supseteq\mathcal{F}(0+)$.
Consequently, $|\mathcal{F}|\ge |\mathcal{F}(0+)|$.
\end{lemma}

\begin{proof}
Suppose that $J\in \mathcal{F}(0+)$.  We prove that $J\in \mathcal{F}$.
From Lemma \ref{lem:general}, for all sufficiently small $\varepsilon>0$,
it holds that $E_J(\varepsilon)\cap K(\varepsilon)\ne\emptyset$, where
$E_J(\varepsilon) = \cap_{i\in J} \partial H_i(\varepsilon)
 = \{x \mid A_J^\top x = b(\varepsilon)\}$
with $A_J=(a_j)_{j\in J}\ (n\times |J|)$.
Let $E_J = E_J(0)$.  From Corollary \ref{co:li},
the column vectors of $A_J$ are linearly independent,
and hence $E_J\ne\emptyset$.

Assume that $E_J\cap K=\emptyset$.
Since both $E_J$ and $K$ are closed polyhedra in $\mathbb{R}^n$,
there is a hyperplane separating $E_J$ from $K$, and the distances
of the hyperplane from $E_J$ and from $K$ are greater than $\delta>0$.
By choosing $\varepsilon>0$ such that
\[
 \max_{i=1,\ldots,m}\frac{\varepsilon}{\Vert a_i\Vert} < \delta,
\]
this hyperplane also separates $E_J(\varepsilon)$ from $K(\varepsilon)$.
Therefore, $E_J(\varepsilon)\cap K(\varepsilon)=\emptyset$, which contradicts
the assumption that $J\in\mathcal{F}(0+)$.
Hence, $E_J\cap K\ne\emptyset$, from which $J\in\mathcal{F}$ follows.
\end{proof}

The following theorem gives another version of the abstract tube associated
with a closed convex polyhedron $K$.
This is an improvement of Corollary 1 of
Naiman and Wynn (1997)\cite{Naiman-Wynn97},
who proved (\ref{at-2}) for almost all $x$ (i.e., a weak abstract tube).

\begin{theorem}
\label{thm:at}
The pair $(\mathcal{H},\mathcal{F}(0+))$
is an abstract tube in the sense that
\begin{equation}
\label{at-2}
\mathbf{1}_{\cup_{i=1}^m H_i^c}(x) = \sum_{J\in\mathcal{F}(0+)}
 (-1)^{|J|-1} \mathbf{1}_{\cap_{i\in J} H_i^c}(x) 
 \quad\mbox{for all $x\in\mathbb{R}^n$}.
\end{equation}
\end{theorem}

\begin{proof}
As in Corollary 1 of Naiman and Wynn (1997)\cite{Naiman-Wynn97},
it can be easily proved that
(\ref{at-2}) holds for $x\notin \cup_{i=1}^m \partial H_i$.
We show that (\ref{at-2}) holds for every $x\in \cup_{i=1}^m \partial H_i$.

If $x\in K \cap \cup_{i=1}^m \partial H_i$,
i.e., $x$ is on the boundary of $K$,
then $x\notin H_i^c$ for all $i$, and (\ref{at-2}) holds as $0=0$.

Suppose that $x\in K^c \cap \cup_{i=1}^m \partial H_i$.
Recall that for sufficiently small $\varepsilon>0$, 
\begin{equation}
\label{at-3}
 \mathbf{1}_{\cup_{i=1}^m H_i(\varepsilon)^c}(x) = \sum_{J\in\mathcal{F}(0+)}
 (-1)^{|J|-1} \mathbf{1}_{\cap_{i\in J} H_i(\varepsilon)^c}(x)
\end{equation}
holds.
Since $K$ is closed and $\varepsilon$ is sufficiently small,
the left side of (\ref{at-3}) becomes
$\mathbf{1}_{\cup_{i=1}^m H_i(\varepsilon)^c}(x)
=\mathbf{1}_{\cup_{i=1}^m H_i^c}(x)=1$.

When $x\in\partial H_i$, $a_i^\top x = b_i \le b_i(\varepsilon)$, and hence
$\mathbf{1}_{H_i^c}(x)=\mathbf{1}_{H_i(\varepsilon)^c}(x)=0$.
When $x\notin\partial H_i$,
$\mathbf{1}_{H_i^c}(x)=\mathbf{1}_{H_i(\varepsilon)^c}(x)$,
because the distance between $x$ and $\partial H_i$ is positive. 
Therefore, in the right side of (\ref{at-3}),
$\mathbf{1}_{\cap_{i\in J} H_i(\varepsilon)^c}(x)
= \prod_{i\in J}\mathbf{1}_{H_i(\varepsilon)^c}(x)
= \prod_{i\in J}\mathbf{1}_{H_i^c}(x)
= \mathbf{1}_{\cap_{i\in J} H_i^c}(x)$.

Summarizing the above, even when
$x\in K^c \cap \cup_{i=1}^m \partial H_i$,
the expression (\ref{at-3}) with $H_i(\varepsilon)^c$ replaced by $H_i^c$
holds, which is nothing but (\ref{at-2}).
\end{proof}

\begin{remark}
\label{rem:improvement}
The abstract tube $(\mathcal{H},\mathcal{F}(0+))$ substantially improves
the abstract tube $(\mathcal{H},\mathcal{F})$ because
\vspace*{-1.5mm}
\begin{itemize}
\item[(i)]
the number of elements of $\mathcal{F}(0+)$ is not greater than
that of $\mathcal{F}$ (Lemma \ref{lem:number}), and
\item[(ii)]
$\cap_{i\in J} H_i^c$, $J\in\mathcal{F}(0+)$,
is a simple cone in $\mathbb{R}^n$,
or is the direct sum of a simple cone and a linear subspace %
 (Corollary \ref{co:li}).  This is not always the case for $\mathcal{F}$.
\end{itemize}
\end{remark}

We present examples of abstract tubes for perturbed polyhedra below.

\begin{example}
Applying a perturbation to the inequality system in Example \ref{ex:pyramid}:
\begin{align*}
-x_1 - x_2 + x_3 \le 1 +\varepsilon^1 & \Hspace [1] \\
-x_1 + x_2 + x_3 \le 1 +\varepsilon^2 & \Hspace [2] \\
+x_1 + x_2 + x_3 \le 1 +\varepsilon^3 & \Hspace [3] \\
+x_1 - x_2 + x_3 \le 1 +\varepsilon^4 & \Hspace [4]
\end{align*}
where $\varepsilon$ is an infinitesimal positive real number.
The corresponding simplicial complex is shown to be
\[
 \mathcal{F}(0+) =
 \{ 1,\, 2,\, 3,\, 4,\, 12,\, 14,\, 23,\, 24,\, 34,\, 124,\, 234 \},
\]
which is a proper subset of $\mathcal{F}$ in (\ref{F}).
The number of terms $|\mathcal{F}(0+)|=11$ is less than $|\mathcal{F}(0+)|=15$.
The maximum number of elements of $J\in \mathcal{F}(0+)$ is $3\,(=n)$.
The difference between $\mathcal{F}$ and $\mathcal{F}(0+)$ is
\[
 \mathcal{F}\setminus\mathcal{F}(0+) =
 \{ 13,\, 123,\, 134,\, 1234 \}.
\]
The members of $\mathcal{F}\setminus\mathcal{F}(0+)$ are characterized as
the sets containing $13$, that is,
the non-existent edge in $K(\varepsilon)$.
\end{example}

\begin{example}
Applying a perturbation to the inequality system in Example \ref{ex:redundant}:
\begin{align*}
 +x_1 - x_2 + 0 x_3 \le 0 +\varepsilon^1 & \Hspace [1] \\
 -x_1 - x_2 + 0 x_3 \le 0 +\varepsilon^2 & \Hspace [2] \\
 0 x_1 - x_2+ 0 x_3 \le 0 +\varepsilon^3 & \Hspace [3]
\end{align*}
The corresponding simplicial complex is shown to be
\[
 \mathcal{F}(0+) = \{ 1,\, 2,\, 3,\, 13,\, 23 \}.
\]
The largest $|J|$, $J\in\mathcal{F}(0+)$, is 2,
as proved in (ii) of Corollary \ref{co:li}.
Note that a different perturbation
obtained by altering the order of inequalities
\begin{align*}
 +x_1 - x_2 + 0 x_3 \le 0 + \varepsilon^2 & \Hspace [1] \\
 -x_1 - x_2 + 0 x_3 \le 0 + \varepsilon^3 & \Hspace [2] \\
 0 x_1 - x_2+ 0 x_3 \le 0 + \varepsilon^1 & \Hspace [3]
\end{align*}
yields a different simplicial complex
\[
 \mathcal{F}(0+) = \{ 1,\, 2,\, 12 \}.
\]
This is an example where $|\mathcal{F}(0+)|$ depends on the perturbation.
\end{example}

Taking the expectation of (\ref{at-2}) with respect to
the normally distributed random vector $x$, we obtain
\[
 1- P(K) = \sum_{J\in\mathcal{F}(0+)}
 (-1)^{|J|-1} P(\cap_{i\in J} H_i^c).
\]
As mentioned in Remark \ref{rem:improvement},
$\cap_{i\in J} H_i^c$ appearing in the right side is
of the class of cones whose probability can be calculated by
the method of Miwa et al.\ (2003)\cite{Miwa-etal03}.

\section{Construction of the abstract tube} 
\label{sec:construction}

In order to construct the abstract tube $(\mathcal{H},\mathcal{F}(0+))$,
we need to determine whether the system
\begin{align*}
& a_i^\top x = b_i(\varepsilon), \ \ i\in J, \\
& a_i^\top x \le b_i(\varepsilon), \ \ i\notin J,
\end{align*}
has a solution for a given subset $J\subseteq \{1,\ldots,m\}$,
where $\varepsilon$ is an infinitesimal positive real number.
For that purpose,
linear programming (LP) techniques for degenerate systems are useful.

First rewrite the problem in a linear programming format as follows:
\begin{problem}
For a given subset $J\subseteq \{1,\ldots,m\}$ determine whether
there exists a feasible solution $x,y\ge 0\ (x,y\in\mathbb{R}^n)$ and
$u_i\ge 0\ (i=1,\ldots,m)$, $v_j\ge 0\ (j\in J)$ such that
\begin{align*}
& u_i + a_i^\top (x-y) - b_i(\varepsilon) = 0, \ \ i=1,\ldots,m, \\
& v_j - a_j^\top (x-y) + b_j(\varepsilon) = 0, \ \ j\in J.
\end{align*}
\end{problem}

To solve this problem
we first define a tableau (matrix) corresponding to the above linear system.
Let $A_J=(a_j)_{j\in J}$ ($n\times |J|$),
$b_J(\varepsilon)=(b_j(\varepsilon))_{j\in J}$ ($|J|\times 1$),
and define an $M\times N$ matrix
\[
[t_{ij}]_{M\times N} =
 \left[\begin{matrix}
   A^\top & -A^\top   & -b(\varepsilon) \\
-A_J^\top &  A_J^\top & b_J(\varepsilon) \\
 1_n^\top & 1_n^\top  & 0 \end{matrix}\right]_{M\times N}
\]
where $M=m+|J|+1$, $N=2n+1$ and $1_n^\top=(1,\ldots,1)_{1\times n}$.
An algorithm for checking the feasibility of $J$ is given as follows
(Theorem 4-4-3 of Iri, 1973\cite{Iri73}).

\begin{algorithm}
\label{alg:pivot}
\begin{itemize}
\item[1.]
Choose an index $i\,(=:i_0)$ such that $i\le M-1$ and $t_{iN}>0$.
If no such $i$ exists, then the system is feasible.

\smallskip
\item[2.]
Among $j$'s such that $j\le N-1$ and $t_{i_0,j}<0$
(if no such $j$ exists, then the system is not feasible),
choose an index $j\,(=:j_0)$ that maximizes $t_{Mj}/|t_{i_0,j}|$. 

\smallskip
\item[3.]
For $i\in \{1,\ldots,M\}\setminus\{i_0\}$ and
$j\in \{1,\ldots,N\}\setminus\{j_0\}$, let
\begin{align*}
t_{i j}     &:= t_{ij} - t_{ij_0}t_{i_0 j}/t_{i_0 j_0}, \\
t_{i j_0}   &:= t_{ij_0}/t_{i_0 j_0}, \\
t_{i_0 j}   &:= - t_{i_0 j}/t_{i_0 j_0}, \\
t_{i_0 j_0} &:= 1/t_{i_0 j_0}.
\end{align*}

\item[4.]
Go to Step 1.
\end{itemize}
In Step 1, $t_{iN}$ is a polynomial in $\varepsilon$.
$t_{iN}>0$ means that $k\,(0\le k\le n)$ exists such that
$t_{iN} = c_k \varepsilon^k + c_{k+1} \varepsilon^{k+1} + \cdots
 + c_n \varepsilon^n$ and $c_k>0$.
\end{algorithm}

Using Algorithm \ref{alg:pivot},
for every nonempty subset $J\subseteq\{1,\ldots,m\}$
we can confirm whether $J$ is feasible (i.e., $J\in\mathcal{F}(0+)$).
The following points should be incorporated into the implementation
of the algorithm.

\begin{remark}
Thanks to Corollary \ref{co:li},
the range of $J$ for checking the feasibility can be restricted to
\begin{equation}
\label{candidate}
 \mathcal{J}_{r,m} =
 \{ J \mid J \subseteq \{1,\ldots,m\},\,0< |J|\le r \},
\end{equation} 
where $r=\mathrm{rank}(A)$.
Moreover, if $J_1 \subset J_2$ and $J_1\notin\mathcal{F}(0+)$,
then $J_2\notin\mathcal{F}(0+)$.
This reduces the range for checking the feasibility.
\end{remark}

\begin{remark}
In Algorithm \ref{alg:pivot} the polynomial
$t_{iN}= c_0 + c_1 \varepsilon^1 + \cdots + c_n \varepsilon^n$ can be
represented as an $n+1$ vector $(c_0,c_1,\ldots,c_n)$.
In this representation $t_{iN}>0$ means $(c_0,c_1,\ldots,c_n)>(0,\ldots,0)$
in the lexicographic order.
This technique is known as the lexicographic method
in linear programming\cite{Iri73,Vanderbei01}.
\end{remark}

\begin{remark}
In Algorithm \ref{alg:pivot} the judgment that a value $z$ is positive
($z>0$) should be implemented as $z>\mathtt{eps}$, where $\mathtt{eps}$ is
a small positive number comparable to the machine epsilon.
For example, the machine epsilon in R (32-bit) is
$2^{-52} \doteq 2\times 10^{-16}$,
and in our prototype R program $\mathtt{eps}$ is set to $10^{-14}$. 
\end{remark}

\begin{remark}
In coding a computer program it is convenient to treat the candidate set
$\mathcal{J}_{r,m}$ in (\ref{candidate}) as a totally ordered set.
Note that
\[
 |\mathcal{J}_{r,m}| = \sum_{j=1}^r {m \choose j}, \quad
 |\mathcal{J}_{0,m}| = 0.
\]
We introduce a total order as follows:
For $J_1=\{ d_1,\ldots,d_{l_1} \}$ ($d_1<\cdots<d_{l_1}$) and
$J_2=\{ e_1,\ldots,e_{l_2} \}$ ($e_1<\cdots<e_{l_2}$)
in $\mathcal{J}_{r,m}$,
define $J_1\preceq J_2$ if $l_1<l_2$, or $l_1=l_2$ and
$(d_1,\ldots,d_{l_1})\le (e_1,\ldots,e_{l_1})$ (lexicographically).
Then $J=\{d_1,\ldots,d_l\}$ ($d_1<\cdots<d_l$) is the $L$-th smallest element
in $\mathcal{J}_{r,m}$, where
\begin{equation*}
 L= 1 + |\mathcal{J}_{l-1,m}|
 + \sum_{j=1}^l \sum_{k=d_{j-1}+1}^{d_j-1} {m-k \choose l-j}.
\end{equation*}
Here we let $d_0=0$, and
$\sum_{k=d_{j-1}+1}^{d_j-1}=0$ if $d_{j-1}+1>d_j-1$.

Conversely, for a given $L$ such that $1\le L\le |\mathcal{J}_{r,m}|$,
the corresponding $L$-th smallest element
$J=\{d_1,\ldots,d_l\}$ ($d_1<\cdots<d_l$) in $\mathcal{J}_{r,m}$
is obtained as follows:
\begin{equation*}
 l := \max\Bigl\{ l\ge 1 \ \big|\ |\mathcal{J}_{l-1,m}| < L \Bigr\},
\end{equation*}
and for $j=1,\ldots,l$, iteratively,
\begin{equation*}
 d_j := \max\biggl\{ d_j \ge d_{j-1}+1 \ \Big|\ %
 \sum_{i=1}^j \sum_{k=d_{i-1}+1}^{d_{i}-1} {m-k \choose l-i} < L
 - |\mathcal{J}_{l-1,m}|
 \biggr\}.
\end{equation*}
\end{remark}

\section{Tukey's studentized range}
\label{sec:studentized}

The multiple comparisons procedure is a typical application where
the multidimensional normal probability in (\ref{PK}) is required.
This is a statistical procedure for combining several statistical tests.
As an example, consider the comparison of $k$ normal means.
Suppose that, for $i=1,\ldots,k$, $X_i$ is independently distributed
according to
a normal distribution with mean $\mu_i$ and variance $\sigma^2_i$,
and consider testing the equality of the means $\mu_i$, $i=1,\ldots,k$.
For $X=(X_1,\ldots,X_k)$, Tukey's studentized range statistic is defined as
\begin{equation}
\label{tukey}
 T(X) = \max_{1\le i<j\le k} \frac{|X_i-X_j|}{\sqrt{\sigma^2_i+\sigma^2_j}}.
\end{equation}
This is a combination of the test statistics
$|X_i-X_j|/\sqrt{\sigma^2_i+\sigma^2_j}$ for testing $\mu_i=\mu_j$.
To determine critical values and power functions, we need
the probability of the event that the statistic (\ref{tukey})
takes a value less than or equal to a threshold $c$. 

The polyhedron $K=\{ x\in\mathbb{R}^k \mid T(x)\le c \}$ is not proper
in the sense that $K$ contains a 1-dimensional linear subspace
$L=\{x\in\mathbb{R}^k \mid x_1=\cdots=x_k \}$.
$K$ is decomposed as $K=K_1\oplus L$, where $K_1$ is
a proper polyhedron.
Fig.\ \ref{fig:studentized} is a picture of $K_1$ when $k=4$
and the $\sigma^2_i$ are all equal.
When $k\ge 4$ and the variances are equal (let $\sigma^2_i\equiv 1$),
the facets of $K$ are not in the general position because
$(X_1-X_2)/\sqrt{2}=c$, $(X_2-X_3)/\sqrt{2}=c$, and
$(X_3-X_4)/\sqrt{2}=c$ imply $(X_1-X_4)/\sqrt{2}=c$.
One can see that 4 facets share a vertex in Fig.\ \ref{fig:studentized}.
\begin{figure}[h]
\begin{center}
\scalebox{0.6}{\includegraphics{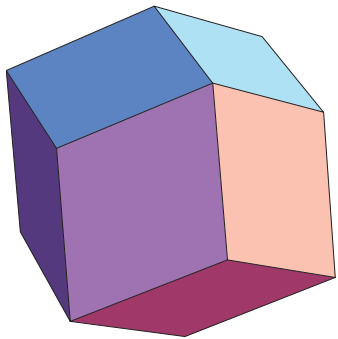}}
\end{center}
\caption{Studentized range ($k=4$)}
\label{fig:studentized}
\end{figure}

Table \ref{tab:studentized} shows the number of $\mathcal{F}(0+)$
for $k=2,\ldots,6$ and constant variances.
Because Algorithm \ref{alg:pivot} depends on the order of the inequalities,
the entries $|\mathcal{F}(0+)|$ in Table \ref{tab:studentized}
may vary when the order changes.
However, numerical experiments suggest that
the numbers $|\mathcal{F}(0+)|$ are independent of the order of inequalities. 
Moreover, even in cases where the variances $\sigma^2_i$ vary, 
$|\mathcal{F}(0+)|$ seems to be unchanged.
\begin{table}
\tbl{Number of terms $|\mathcal{F}(0+)|$ in the abstract tubes}{
\begin{tabular}{ccccccc} \toprule
$k$                 & 2 &  3 &  4 &   5 &    6 \\ \colrule
$m=k(k-1)$          & 2 &  6 & 12 &  20 &   30 \\
$|\mathcal{F}(0+)|$ & 2 & 12 & 62 & 320 & 1682 \\ \botrule
\end{tabular}
}
\begin{tabnote}
$m$ is the number of inequalities (facets).\\
\end{tabnote}
\label{tab:studentized}
\end{table}

In the statistical community there was a conjecture called
the ``Tukey-Kramer conjecture'' which is stated as follows:
Let $F_k(\sigma_1,\ldots,\sigma_k;c)=\Pr(T(X)\le c)$,
where $X\sim N_k(0,\mathrm{diag}(\sigma^2_i))$.
Then for any $k$ and $c$,
$F_k(\sigma_1,\ldots,\sigma_k;c)$ takes its minimum when the variances
$\sigma^2_i$ are all equal.
This conjecture was affirmatively proved by Hayter (1984)\cite{Hayter84}.
Fig.\ \ref{fig:tk} depicts the value $F_5(\sigma_1,\ldots,\sigma_5;c)$
with $\sigma_i^2=(10^s)^{(i-1)/4}$, $i=1,\ldots,5$,
for the range $s\in [-5,5]$.
The constant $c$ is chosen such that $F_5(1,\ldots,1;c)=0.95$.
The minimum is certainly attained at $s=0$, i.e., $\sigma_i^2\equiv 1$.

Finally, note that although the polyhedron $K$ depends on the threshold $c$,
the shape of $K$ is invariant for all $c>0$.
Once an abstract tube is obtained for some $c>0$, this abstract tube
is available for computing the distribution function $\Pr(T(X)\le c)$
for any $c$ in the null and non-null cases.

\begin{figure}[h]
\begin{center}
\vspace*{-5mm}
\scalebox{0.6}{\includegraphics{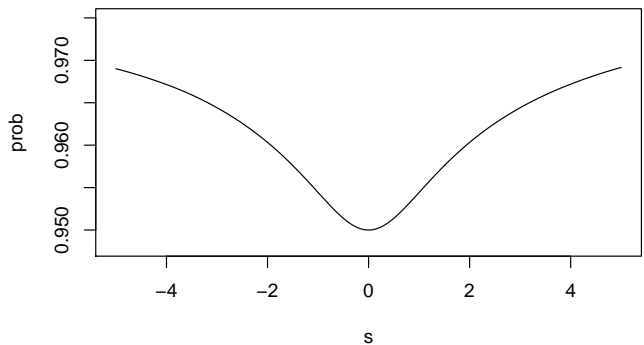}}
\end{center}
\vspace*{-5mm}
\caption{Tukey-Kramer conjecture ($k=5$)}
\label{fig:tk}
\end{figure}

\bibliographystyle{ws-procs9x6}
\bibliography{pivot}

\end{document}